\documentclass[12pt]{article}
\usepackage{amsmath,amssymb}
\usepackage{type1cm}
\usepackage{amsthm}
\usepackage[dvipdfmx]{graphicx}
\RequirePackage[l2tabu, orthodox]{nag} 
\usepackage[all, warning]{onlyamsmath}
\usepackage[subrefformat=parens]{subcaption}
\usepackage{caption}
\usepackage{bm}
\usepackage{authblk}
\usepackage{url}
\usepackage{color}
\newtheorem{theorem}{Theorem}

\newcommand{\E}[2][]{\mathbb{E}_{#1}\mathopen{}\left[#2\mathopen{}\right]}
\newcommand{\Ec}[3][]{\mathbb{E}_{#1}\mathopen{}\left[#2 \middle| #3\mathopen{}\right]}
\newcommand{\mis}{\mathrm{mis}}
\newcommand{\obs}{\mathrm{obs}}
\newcommand{\normal}[2]{\mathrm{N}\left(#1, #2\right)}

\title{Identification and Estimation of Heterogeneous Treatment Effects under Non-compliance or Non-ignorable assignment}
\author[1,2]{Keisuke Takahata}
\author[1,2]{Takahiro Hoshino\thanks{hoshino@econ.keio.ac.jp}}
\affil[1]{Keio University, Japan}
\affil[2]{RIKEN Center for Advanced Intelligence Project, Japan}
\date{}

\begin{document}

\setlength{\abovedisplayskip}{4pt}
\setlength{\belowdisplayskip}{4pt}
\setlength{\baselineskip}{7.1mm}

\maketitle
\thispagestyle{empty}
\begin{abstract}
We provide sufficient conditions for the identification of the heterogeneous treatment effects, defined as the conditional expectation for the differences of potential outcomes given the untreated outcome, under the nonignorable treatment condition and availability of the information on the marginal distribution of the untreated outcome. These functions are useful both to identify the average treatment effects (ATE) and to determine the treatment assignment policy. 
The identification holds in the following two general setups prevalent in applied studies: (i) a randomized controlled trial with one-sided noncompliance and (ii) an observational study with nonignorable assignment with the information on the marginal distribution of the untreated outcome or its sample moments. To handle the setup with many integrals and missing values, we propose a (quasi-)Bayesian estimation method for HTE and ATE and examine its properties through simulation studies. We also apply the proposed method to the dataset obtained by the National Job Training Partnership Act Study.
\end{abstract}



\newpage

\pagestyle{plain}
\setcounter{page}{1}

\section{Introduction}

\subsection{Heterogeneous treatment effects}

Both for randomized controlled trials (RCTs) and observational studies, the average treatment effect (ATE), average treatment effect on the treated (ATT), and average treatment effect on the untreated (ATU) are effects of interest (Rubin, 1974). However, even for an RCT, without perfect compliance with the assigned treatments, only the local average treatment effect (LATE), which is different from ATE, is consistently estimated under the additional conditions (Angrist, Imbens and Rubin, 1996). In an observational study, strong ignorability condition, which requires an assignment to be independent of the potential outcomes given the covariates, is known to play a significant role in the identification of those effects (Rosenbaum and Rubin, 1983).

Moreover, at times, researchers want to identify more individualized or heterogeneous causal effects, while, as their names suggest, ATE, ATT, and ATU are averaged effects over a population or a subset of a population. Estimation of various kinds of heterogeneous effects has received great attention in recent years, particularly in marketing and medicine where personalized treatments are effective, or in policy-making, where the cost (and thus the number of the targets) for special interventions such as job training programmes is limited (Kitagawa and Tetenov, 2018). 

Previously proposed heterogeneous causal effects are functions of the observable variables (e.g. Wager and Athey, 2018); however in this paper we discuss the identification and inference of heterogeneous treatment effects (HTEs), which we define in this paper as 
\begin{align*}
\mathrm{HTE}(y_{0}) = \Ec{y_{1}-y_{0}}{y_{0}}, 
\end{align*}
where $y_{1}\in \mathbb{R}$ and $y_{0}\in \mathbb{R}$ are the potential outcome variable under the (special) treatment condition (with higher cost) and the (default) control condition respectively, and $x \in \mathbb{R}^{d}$ is a $d$-dimensional covariate vector. HTE is a function of $y_{0}$, which can indicate how much effect the unit whose outcome is $y_{0}$ under the untreated condition, would get if the unit is assigned to the treatment condition.

It is often useful to estimate this function because in many real applications special treatment provided with high monetary cost or mental burden is not always effective. For example, it is well known that monetary reward for enhancing physical activity can promote exercise in individuals without good exercise habits but may undermine the intrinsic motivation for individuals with good exercise habits (e.g., Deci, 1971; Charness and Gneezy, 2009), thus providing opposite effects on physical activity in different types of individuals. As will be mentioned and re-analysed in Section 5, the effect of a well-known job training provided under the Job Training Partnership Act (JTPA) depends on trainees' demographic variables (Bloom et al., 1997; Abadie et al., 2002). Moreover, as will be mentioned later, the causal estimands usually considered in the previous studies, ATE, ATT and ATU, are expressed as the expectations of the HTE function.

\subsection{Identification problem and the main result}

Although HTEs may have implications that attract researchers, the identification of HTEs is not trivial owing to the dependence of the potential outcome variable: we need to identify the density of $y_{1}$ given $y_{0}$, $p(y_{1}|y_{0})$, but $y_{1}$ and $y_{0}$ are never observed simultaneously. Then, we need to deal with the missing mechanism which is nonignorable (or missing not at random, Little and Rubin, 2002). Under nonignorable missingness, indentifiability is not assured even under full parametric assumptions (e.g. Miao et al., 2016; Cui et al., 2017). Therefore, additional conditions are needed for identification. 

To this end, we introduce the following two assumptions in this paper.
First, we consider relaxing strong ignorability condition as
\begin{align*}
p(z|y_{1}, y_{0},x) = p(z|y_{0},x), 
\end{align*}
where $z \in \{0,1\}$ is an indicator which is $z=1$ when $y_{1}$ is observed (i.e. when assigned and complying with the treatment condition). We refer to this assumption as \textit{weak ignorability}. This assumption is justifiable for the following two reasons. First, it is always weaker than strong ignorability assumption. Second, since $z$ precedes the outcome in causal inference, it is natural to assume that assignment of or compliance with the treatment will be influenced by the default value of the outcome, $y_{0}$, rather than by the outcome under some special treatment, $y_{1}$ (Hoshino, 2013).
Although it is not straightforward to observe how weak ignorability works in the identification of HTE, the details are described in Section 2. Second, we assume that the information on the distribution of the untreated outcome $p(y_{0})$ or its moments is available. We present two practical examples where the second assumption is satisfied in the next subsection.

Our main result is that, under weak ignorability and the availability of the information on the distribution of the untreated outcome $y_{0}$, it is sufficient for the identification of $p(y_{1}|y_{0},x)$ and the HTE that the extended propensity score (described in section 2.1, Eq.~(\ref{eq:hirano_result})) (i) is specified as the logistic regression, (ii) is specified so that certain additivity holds between $y_{0}$ and $x$, and (iii) has the linear term of $y_{0}$ in the regression function. This condition also assures the identification of causal estimands such as ATT, ATE, and ATU, while under the traditional setup with one-sided non-compliance only ATT is identified (see Section 2 for detail).

\subsection{Two setups considered in this paper}

As typical examples where $p(y_{0})$ can be obtained, we consider the following two setups prevalent in applied studies (see also Figure~\ref{fig:fig1}): (a) randomized controlled trials which are conducted with imperfect compliance, called ``one-sided noncompliance''  (Imbens and Rubin, 2015), in that for the control group all the participants comply with the control condition while for the treatment group not all the participants comply with their treatment, or individuals are allowed to choose their treatment, and (b) observational studies in which external information on the population or a random sample of the population is available.

Setup (a) is sometimes called a ``randomized consent design'' (Zelen, 1979, 1990). In this setup, $r$ is an assignment indicator which is $r=1$ if the target unit is assigned to the treatment condition, while $z$ is an treatment (or more strictly, compliance to treatment) indicator which is $z=1$ when the unit complies with assignment to the treatment and $z=0$ when it does not. Note that $y_{1}$ and $z$ are missing for $r=0$.
In setup (b), $r=1$ if the units belong to an observational study in which $z$ is a treatment indicator not determined by the researchers.
For examples of setup (b), there is a survival-time distribution for the control condition in a population, often available in medical research, or an income distribution estimated using census data in economics. In marketing, when a company tries to evaluate the effect of a new marketing promotion for specific customers, the company can choose a random sample ($r=1$) and apply the promotion to targeted customers ($z=1$) while the distribution of various variables for customer ($r=0$) is known.

In this paper, we assume that $r$ is independent of potential outcomes $y_{1}$ and $y_{0}$ (i.e. random assignment for case (a) and random sampling for case (b)); however, the results obtained in this paper can be easily generalized to the case when $p(r|y_{1},y_{0},x)=p(r|x)$.

\begin{figure}[htbp]
	\centering
	\subcaptionbox{RCT with one-sided non-compliance\label{fig:fig1a}}
	{\includegraphics[width=0.45\linewidth]{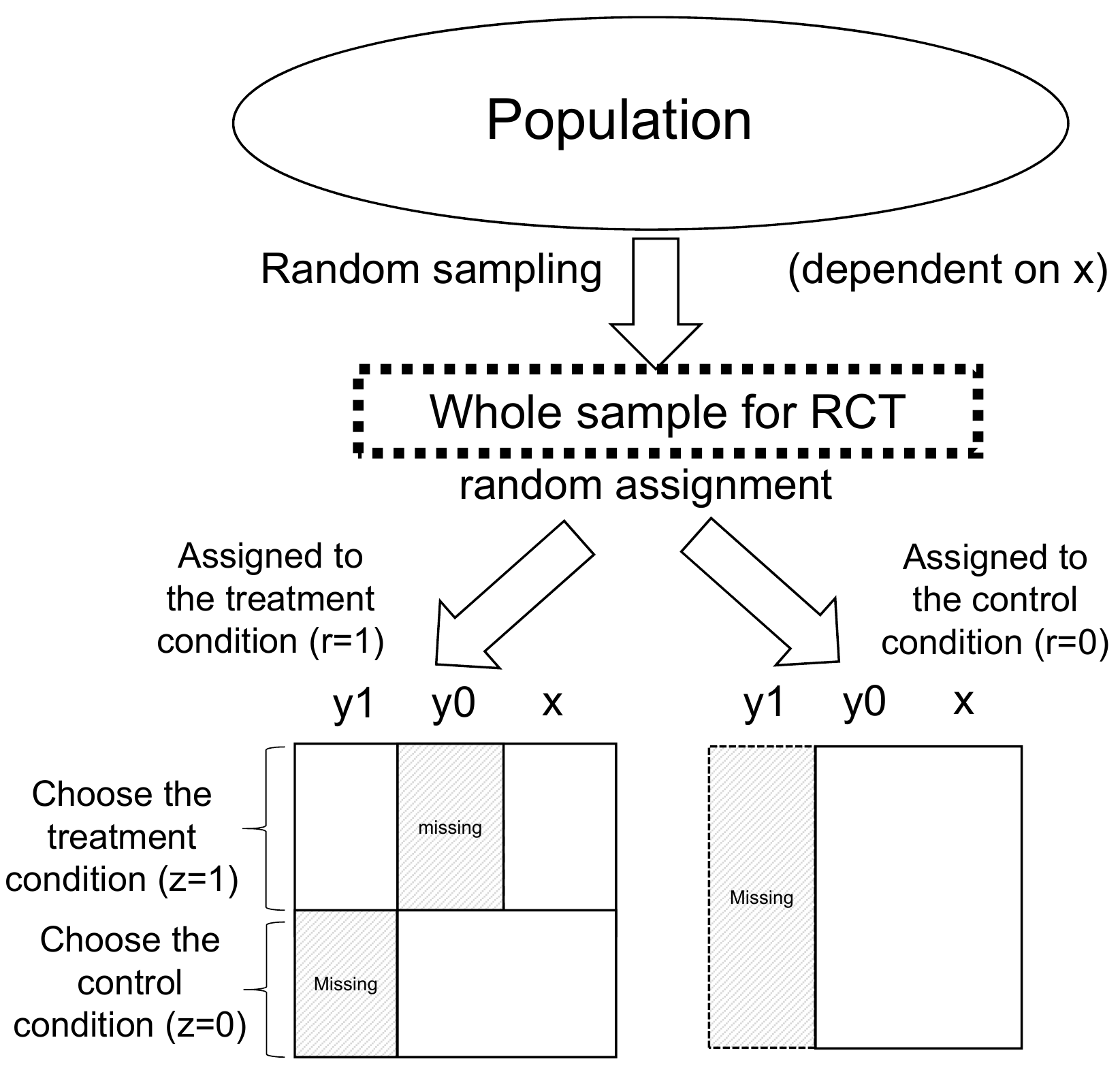}} 
	\subcaptionbox{Observational study with population information\label{fig:fig1b}}
	{\includegraphics[width=0.45\linewidth]{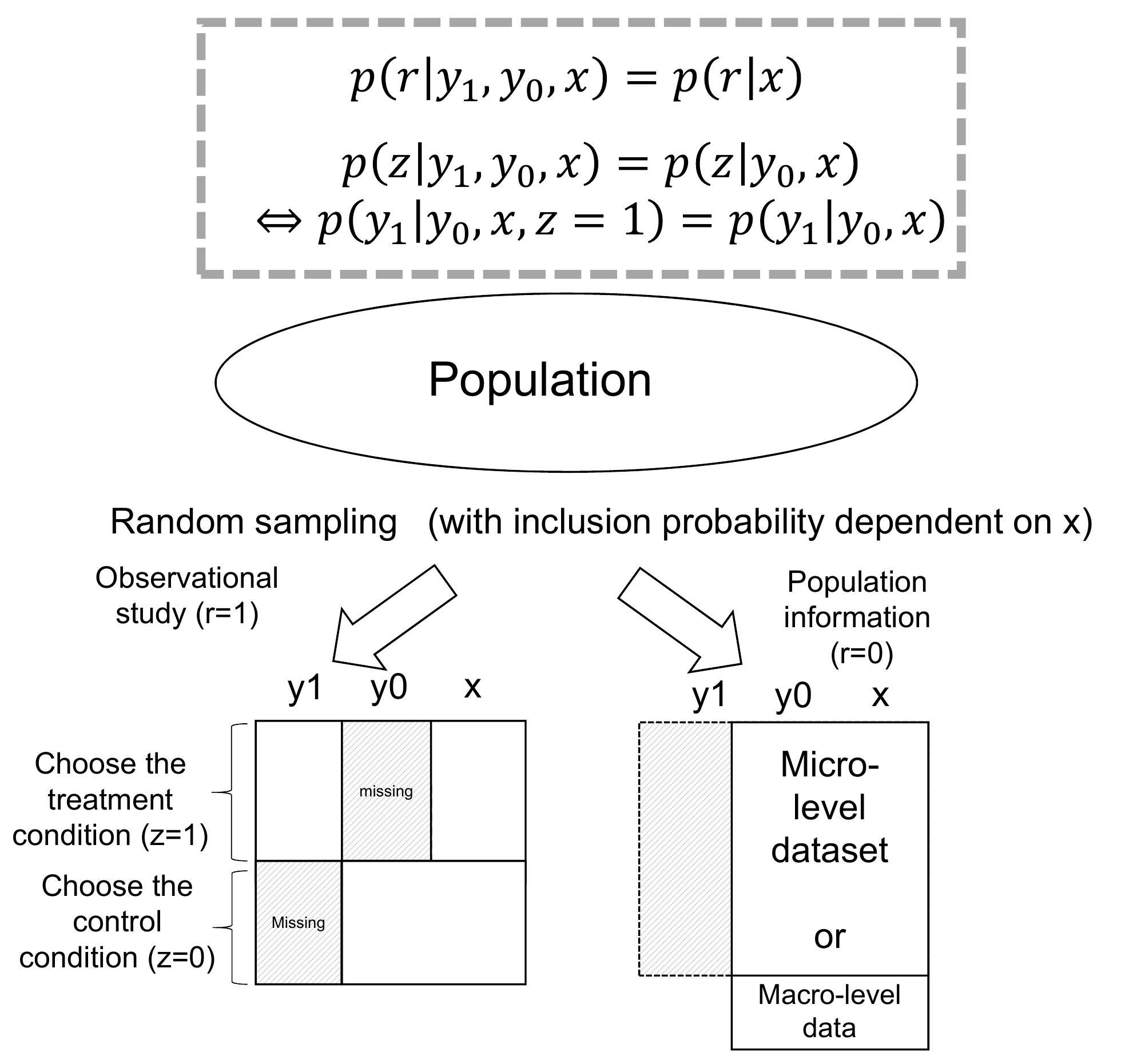}} 
	\caption{The two setups we considered in this paper}
	\label{fig:fig1} 
\end{figure}

\subsection{Organization of the paper}

The remainder of this paper is organized as follows. In Section 2, the HTE identification problem is formulated and the sufficient condition for identification is derived. In Sections 3 we propose a Bayesian and quasi-Bayesian method to estimate the parameters of interest: the HTE function, ATE, ATT, and ATU. Note that the purpose of employing the Bayesian framework here is to use a data augmentation approach to tackle missing values and integrals, and it is not necessary to use informative prior distribution to obtain estimates of them under the conditions obtained in Section 2.
Section 4 illustrates the performance of the proposed estimators through simulation studies. In Section 5, we apply the proposed method to the dataset obtained by the National Job Training Partnership Act (JTPA) Study, a typical randomized controlled trial with one-sided non-compliance.
Finally, the connections to several related fields are discussed and a conclusion is presented in Section 6.

\section{Identification of Causal Estimands}
In this section we discuss the sufficient conditions for identification of the causal estimands.
The assumption for identification of the causal estimands and the estimation framework is the same for setups (a) and (b) in Figure~\ref{fig:fig1}, except for the case when micro-level information is not available but macro-level information is available in setup (b); under the macro-level information alone, the assumption is more restrictive and the estimation method is more complicated (see this section and Section 3 for detail).
First we assume that we obtain at least the information on the marginal distribution of the untreated outcome, $p(y_{0})$, which is available for case (b) considered in the introduction (see also Figure~\ref{fig:fig1}). Later in this section, we will consider case (a), when we can obtain the information of $p(y_{0},x)$, in which the assumptions will be more relaxed.

Our discussion is founded on Hirano et al.~(2001) and Newey and Powell~(2003). Hirano et al.~(2001) consider a situation of nonignorable attrition in a two-period panel while refreshment samples, which are new additional units randomly sampled from the target population, are available (see also Nevo, 2003; Chen et al., 2017). They provide the sufficient condition for identification of the nonignorable missing mechanism. Given Hirano et al.~(2001)'s result, we observe that the identifiability of the HTE reduces to the uniqueness of a solution of some integral equation. The integral equation to be solved here has the same structure as that of nonparametric instrumental variable models; the uniqueness of the integral equation is thus discussed based on Newey and Powell~(2003), who characterize the uniqueness of the integral equation as the completeness of a certain conditional distribution.

\subsection{General identification conditions}

HTE, ATE, ATT and ATU are formulated as
\begin{align*}
\begin{split}
\mathrm{HTE}(y_{0}) &= \Ec{y_{1}-y_{0}}{y_{0}} = \E[x|y_{0}]{\Ec{y_{1}}{y_{0},x}} - y_{0}, \\
\mathrm{ATE} &= \E{y_{1}-y_{0}} = \E[y_{0},x]{\Ec{y_{1}}{y_{0},x}} - \E{y_{0}},\\
\mathrm{ATT} &= \Ec{y_{1}-y_{0}}{z=1} = \Ec{y_{1}}{z=1} - \Ec{y_{0}}{z=1}, \\
\mathrm{ATU} &= \Ec{y_{1}-y_{0}}{z=0} = \E[y_{0},x]{\Ec{y_{1}}{y_{0},x}} - \Ec{y_{0}}{z=0},
\end{split}
\end{align*}
where $\E[u]{\cdot}$ denotes the expectation over $u$.
Note here that $z$ is the compliance indicator in setup (a) but the treatment indicator in setup (b) in Figure~\ref{fig:fig1}.
We can consistently estimate $\E{y_{1}|z=1}$ by using the observed data and $\E{y_{0}}$ by the distributional information. Moreover, as $p(y_{0}|z=1)$ can be calculated by $p(y_{0}|z=1) = (p(y_{0})-p(y_{0}|z=0)p(z=0))/p(z=1)$, the identifiability of ATT is trivial. Therefore, it suffices to provide the conditions for the identification of $p(y_{0},x)$ and $p(y_{1}|y_{0}, x)$ for ATE and ATU, and this is clearly also sufficient for HTE.

First, we discuss the identification of $p(y_{0},x)$. Following Bayes' rule, $p(y_{0},x)$ can be written as
\begin{align}
p(y_{0},x) = \frac{p(y_{0},x|z=0)p(z=0)}{p(z=0|y_{0},x)} = \frac{p(y_{0},x|z=0)p(z=0)}{1-p(z=1|y_{0},x)}.
\label{eq:hirano_bayes}
\end{align}
$p(z=1|y_{0},x)$ can be interpreted as the extended version of the propensity score, where the original version of the propensity score is generally defined as the probability of being assigned to the treatment group given only the covariates (Rosenbaum and Rubin, 1983).
We provide the sufficient condition for the identification of $p(z=1|y_{0},x)$ when $p(y_{0})$ is known, by an application of the result of Hirano et al.~(2001).
\begin{theorem}
If $p(y_{0})$ is known, the extended propensity score is identified in the form of
\begin{align}
p(z=1|y_{0},x) = g(k_{0}+k_{y_{0}}(y_{0}) + k_{x}(x)), 
\label{eq:hirano_result}
\end{align}
where $g$ is a known function that is differentiable, strictly increasing with $\lim_{x\to -\infty} g(x) = 0$ and $\lim_{x\to \infty} g(x) = 1$ and $k_{y_{0}}(\cdot)$ and $k_{x}(\cdot)$ are unique sets of functions subject to normalization, $k_{y_{0}}(0) = k_{x}(0) = 0$.
\end{theorem}

\begin{proof}
See p. 1653 of Hirano et al. 2001.
\end{proof}

\noindent
The point of Theorem 1 is that the extended propensity score must be specified so as not to include an interaction term between $y_{0}$ and $x$; that is, additivity has to hold in eq.~(\ref{eq:hirano_result}). Given this result, it is straightforward that $p(y_{0},x)$ is identifiable.

Next, we discuss the identification of $p(y_{1}|y_{0}, x)$.
Let us consider the following integral equation:
\begin{align}
\notag p(y_{1}|x,z=1) &= \int p(y_{1}|y_{0},x,z=1)p(y_{0}|x,z=1)dy_{0} \\
&= \int p(y_{1}|y_{0}, x)p(y_{0}|x,z=1)dy_{0}, 
\label{eq:int_eq}
\end{align}
where the second equality holds under weak ignorability. Note that we can consistently estimate $p(y_{1}|x,z=1)$ using the observed data. Moreover, by substituting $z=0$ with $z=1$ in the middle part of eq.~(\ref{eq:hirano_bayes}), it is easily verified that $p(y_{0}|x,z=1)$ is identifiable. Hence, if there is a unique $p(y_{1}|y_{0}, x)$ that satisfies eq.~(\ref{eq:int_eq}), then $p(y_{1}|y_{0}, x)$ is identifiable, leading to the identifiability of HTE, ATE, and ATU.

Eq.~(\ref{eq:int_eq}) is called the Fredholm integral equation of the first kind; these types of equations are known to be ill-posed problems, in which additional conditions or regularizations are needed to obtain a stable and unique solution. Several models, such as nonparametric instrumental variable models and measurement error models, are known to reduce to solving this type of equation and have been studied in econometrics as statistical inverse problems (see e.g.Carrasco et al., 2007; Horowitz, 2009). Newey and Powell~(2003) characterize the uniqueness of a solution of this type of integral equation as the completeness of the distribution on which the expectation of the function of interest is based; since their pioneering work, the completeness condition has been widely used in econometrics (e.g. Hall and Horowitz, 2005; Blundell et al., 2007; Darolles et al., 2011; Horowitz, 2011; Horowitz and Lee, 2012). Following Newey and Powell~(2003)'s discussion, we provide the sufficient condition for the identification of HTE, ATE, and ATU.
\begin{theorem}
Under weak ignorability and if: 
\begin{itemize}
\item[(c.1)] $p(y_{0})$ is known;
\item[(c.2)] $p(z=1|y_{0},x)$ has no interaction term between $y_{0}$ and $x$;
\item[(c.3)] $p(z=1|y_{0},x)$ is specified by the logistic function, $g(t)=\frac{1}{1+\exp(-t)}$;
\item[(c.4)] $k_{y_{0}}(\cdot)$ in eq.~(\ref{eq:hirano_result}) includes the linear term: $k_{y_{0}}(y_{0})=\theta_{y_{0}}y_{0}+\bar{k}_{y_{0}}(y_{0})$, where $\theta_{y_{0}}$ is non-zero and bounded, and $\bar{k}_{y_{0}}(y_{0})$ is any function of $y_{0}$ subject to normalization, $\bar{k}_{y_{0}}(0)=0$;
\end{itemize}
then, HTE and ATE and ATU are identifiable.
\end{theorem}

\begin{proof}
Let $p(y_{1}|y_{0}, x)$ and $\tilde{p}(y_{1}|y_{0},x)$ be any solutions of eq.~(\ref{eq:int_eq}). By subtracting both equations with these solutions inserted, we obtain
\begin{align}
\notag &\int (p(y_{1}|y_{0}, x) - \tilde{p}(y_{1}|y_{0},x) )p(y_{0}|x,z=1)dy_{0} \\
= &\int h(y_{0},y_{1},x) p(y_{0}|x,z=1)dy_{0} = 0, 
\label{eq:completeness}
\end{align}
where $h(y_{0},y_{1},x) = p(y_{1}|y_{0}, x) - \tilde{p}(y_{1}|y_{0},x)$. By Theorem 1, $p(y_{0}|x,z=1)$ is identified under condition (c.1) and (c.2). Then if it is the case that $h(y_{0},y_{1},x)$, which solves eq.~(\ref{eq:completeness}), is always zero for all $y_{0}$ given any fixed $y_{1}$ and $x$, then $p(y_{1}|y_{0}, x)$ is equal to $\tilde{p}(y_{1}|y_{0},x)$ and eq.~(\ref{eq:int_eq}) has a unique solution. This fact implies that $p(y_{1}|y_{0}, x)$ is identifiable if $p(y_{0}|x,z=1)$ is complete.

Under condition (c.3), $p(y_{0}|x,z=1)$ can be rewritten as
\begin{align}
\notag p(y_{0}|x,z=1) &= \frac{p(y_{0},x)p(z=1|y_{0},x)}{p(x,z=1)} = \frac{p(y_{0},x|z=0)p(z=0)p(z=1|y_{0},x)}{p(z=0|y_{0},x)p(x,z=1)} \\
&= \exp(k_{0} + k_{y_{0}}(y_{0}) + k_{x}(x)) \frac{p(y_{0},x|z=0)p(z=0)}{p(x,z=1)}.
\label{eq:target_temp}
\end{align}
In addition, under (c.4), plugging eq.~(\ref{eq:target_temp}) into eq.~(\ref{eq:completeness}) yields
\begin{align}
\int \exp(\theta_{y_{0}}y_{0}) u(y_{0},y_{1},x) dy_{0} = 0
\label{eq:result}
\end{align}
where
\begin{align}
u(y_{0},y_{1},x) = h(y_{0},y_{1},x)\exp(k_{0}+\bar{k}_{y_{0}}(y_{0})+k_{x}(x)) \frac{p(y_{0},x|z=0)p(z=0)}{p(x,z=1)}.
\label{eq:def_m}
\end{align}

Eq.~(\ref{eq:result}) has the same structure as in the discussion of the completeness of the exponential family. Therefore, considering Theorem 1 (Lehmann, 1986, p. 142), this leads to $u(y_{0},y_{1},x) = 0$. As each value on the right-hand side of eq.~(\ref{eq:def_m}) is strictly positive except $h(y_{0},y_{1},x)$, it follows that $h(y_{0},y_{1},x) = 0$, that is, eq.~(\ref{eq:int_eq}), has a unique solution.
\end{proof}

Note that, in a case where we only know some of the moments of $p(y_{0})$, two additional conditions need to be imposed for identifying $p(z=1|y_{0},x)$ (Nevo, 2002). First, the moment information includes $\E{y_{0}}$. Second, $\bar{k}(y_{0})$ is a linear combination of known functions of $y_{0}$ whose moments are known. In other words, $k_{y_{0}}(\cdot)$ is specified as $k_{y_{0}}(y_{0})=\sum_{j=1}^{p} \theta_{y_{0}j}T_{j}(y_{0})$, where $T_{1}(y_{0})=y_{0}$ and $\E{T_{j}(y_{0})}$ are known for all $j=1,\dots,p$. 
On the other hand, if we can use the information on $p(y_{0}, x)$ (e.g. the setup (a) given in Section 1), we can relax the additivity condition (c.2) by making suitable modifications on condition (c.4).

Note also that the conditional density $p(y_{1}|y_{0}, x)$ is useful to determine the treatment assignment policy or rule indexed by the decision set $G \subset {\it X} \subset \mathbb{R}^{dim(X)}$ of covariate value $X$ which maximizes the social welfare (Kitagawa and Tetenov, 2018),
\begin{align}
W(G) = \iint_{X \in G} y_1 p(y_1|y_0,x) p(y_0|x) p(x) dy_1 dy_0 dx + \iint_{X \notin G} y_0 p(y_0|x) p(x) dy_0 dx.
\end{align}

\subsection{RCT allowing one-sided non-comliance}

In this subsection we mention some notes for the RCT with one-sided non-compliance (Figure~\ref{fig:fig1a}).
In this setup, we can identify HTE, ATE, ATT, or ATU for a randomized controlled trial with one-sided non-compliance, that is, where for the control group all the participants comply with the control condition while for the treatment group not all the participants comply with their treatment or individuals are allowed to choose their treatment.

The RCT design with one-sided non-compliance is very easy for researchers to implement, because they do not have to enforce complete compliance, and a large portion of RCT studies result in this design, as will be considered in Section 5.

Note that under the condition that $z \neq 1$ for $r=0$ (i.e. one-sided non-compliance), LATE (Imbens and Angrist, 1994; Angrist, Imbens and Rubin, 1996) is equivalent to ATT (or to be more specific, to the average treatment effect for the complier) under the ``monotonicity assumption", then the standard instrumental variable method (often called the Wald estimator) can consistently estimate ATT, but not ATE, ATU, or HTE.
In contrast, under the identification conditions developed here, the proposed method can identify ATE, ATT, ATU, HTE, and these estimands as functions of the covariate vector, without the monotonicity assumption required in the LATE framework.

\section{Estimation}
In this section, we propose an estimation method for $\E{y_{1}|y_{0}}$. 
We consider a case in which the joint distribution of two potential outcomes given the covariates is parametrically modelled: $p(y_{1},y_{0}|x,\psi)=p(y_{1}|y_{0},x,\psi)p(y_{0}|x,\psi)$, where $\psi$ is the parameter vector.

Let us consider a random sample with size $N$ in which for the $i$-th unit the potential random variable is the vector $(y_{i1},y_{i0},x_{i},r_{i},z_{i})$ for $i=1,\dots, N $. In this setup, $z_{i}$ is observed for $r_{i}=1$ but $z_{i}$ and $y_{i1}$ are not observed for $r_{i}=0$. 

For $z_{i}=1$ and $ r_{i}=1$, the likelihood includes integration with respect to $y_{0}$ because of the dependence on $y_{0}$ in the missing mechanism. In general, such integral is not solved analytically, so a numerical calculation is needed. In Bayesian inference we can employ the data augmentation algorithm and sample the missing portion of $y_{0}$, $y_{0}^{\mis}$, as the parameters.

\subsection{Setup (a) and setup (b) with micro-level data}

Under weak ignorability, the likelihood given the observed data for setup (a) or setup (b) with micro-level data is expressed as follows:
\begin{align*}
&L(\psi| y^{\obs}, z^{\obs}, r, x) \\
&= \prod_{i=1}^{N} \Bigl( \left[p(y_{i1},z_{i}=1|x_{i},\psi) \right]^{z_{i}} \left[p(y_{i0},z_{i}=0|x_{i},\psi) \right]^{1-z_{i}} \Bigr)^{r_{i}} \times \Bigl( p(y_{i0}|x_{i},\psi) \Bigr)^{1-r_{i}}  \\
&=\prod_{i=1}^{N} \Bigr( \left[ \int p(y_{i1},y_{i0}|x_{i},\psi)p(z_{i}=1|y_{i0},x_{i},\psi)dy_{i0} \right]^{z_{i}} \\
&\quad \times \left[ p(y_{i0}|x_{i},\psi)p(z_{i}=0|y_{i0},x_{i},\psi) \right]^{1-z_{i}} \Bigr)^{r_{i}} \times \Bigl( p(y_{i0}|x_{i},\psi) \Bigr)^{1-r_{i}},
\end{align*}
where $(y^{\obs},z^{\obs}=(y_{1}^{\obs}, y_{0}^{\obs},z^{\obs})$ is the observed portion of $(y,z)=(y_{1},y_{0},z)$.

The resulting posterior distribution\footnote{In a simple model setup we can employ numerical integration, as shown in the next section.} is
\begin{align}
\label{eq:posterior}
p(\psi| y^{\obs}, z, x) \propto L(\psi| y^{\obs}, z^{\obs} z, r, x) \times p(\psi),
\end{align}
where $p(\psi)$ is the prior distribution of the parameter vector $\psi$.

Alternatively, we can employ a data augmentation approach in which the missing $y_{0}$'s for $r=1$ and $z=1$ are drawn from the joint posterior distribution,
\begin{align}
\label{eq:posteraug}
& p(\psi,y_{0}^{\mis}| y^{\obs}, z^{\obs}, x) \propto p(\psi) \times \\
\notag
& \prod_{i=1}^{N} \Bigl( p(z_{i}|y_{i0},x_{i},\psi) \left[ p(y_{i1},y_{i0}|x_{i},\psi) \right]^{z_{i}} 
\left[ p(y_{i0}|x_{i},\psi) \right]^{1-z_{i}} \Bigr)^{r_{i}}
\times \Bigl( p(y_{i0}|x_{i},\psi) \Bigr)^{1-r_{i}} .
\end{align}

\subsection{Setup (b) with macro-level data}

For setup (b) with macro-level moment information, we must also incorporate auxiliary information into the likelihood to identify the extended propensity score.
Nevo (2003) proposes a GMM-type estimator for the nonignorable missing model with moment conditions by adopting the frequentist approach. To follow Nevo (2003)'s method but from a Bayesian perspective, we employ a quasi-Bayesian approach. Chernozhukov and Hong~(2003) show that if an estimator is the solution of some optimization problem, we can obtain the estimator by using samples from $\exp(Q(\psi))/\int \exp(Q(\psi))p(\psi)d\psi$, where $\psi$ is the parameter of interest and $Q(\psi)$ is the objective function to be minimized. Then, for example, a GMM estimation can be conducted based on MCMC samples. Basically, if we want to use both a likelihood and an objective function, the likelihood has to be transformed to the score function and combined with the objective function into the moment condition, which makes it difficult to draw samples. Recently, Hoshino and Igari (2017) show that the quasi-Bayes estimator is consistent and asymptotically normally distributed if the quasi-posterior is set to be proportional to the likelihood part multiplied by the exponential of the GMM objective function (see also Igari and Hoshino, 2018).
Based on their results, we incorporate macro-level moment information on $p(y_{0})$ into the likelihood, expressed as
\begin{align}
\notag
&L(\psi| y^{\obs}, z^{\obs}, r, x) = \prod_{i:r_{i}=1 }^{N} \left[p(y_{i1},z_{i}=1|x_{i},\psi)\right]^{z_{i}}\left[p(y_{i0},z_{i}=0|x_{i},\psi)\right]^{1-z_{i}} \\
\notag
&=\prod_{i:r_{i}=1 }^{N} \left[ \int p(y_{i1},y_{i0}|x_{i},\psi)p(z_{i}=1|y_{i0},x_{i},\psi)dy_{i0} \right]^{z_{i}} \\
&\quad \times \left[ p(y_{i0}|x_{i},\psi)p(z_{i}=0|y_{i0},x_{i},\psi) \right]^{1-z_{i}} ,
\label{eq:likelihood_int2}
\end{align}
as the quasi-joint posterior distribution:
\begin{align*}
p(\psi,y_{0}^{\mis}| y^{\obs}, x) \propto L(\psi| y^{\obs}, z^{\obs} z, r, x) \times \exp(Q_{0}(\psi)) \times p(\psi),
\end{align*}
where $Q_{0}(\psi)$ is the GMM objective function,
\begin{align}
Q_{0}(\psi) &= -\frac{N_{0}}{2}\left(\frac{1}{N_{0}}\sum_{i:r_{i}=1,z_{i}=0}m_{0}(y_{i0},x_{i},\psi)\right)'W_{0}\left(\frac{1}{N_{0}}\sum_{i:r_{i}=1,z_{i}=0}m_{0}(y_{i0},x_{i},\psi)\right),
\label{eq:obj_gmm}
\end{align}
$N_{0}=\sum_{i=1}^{N} (1- r_{i})$, $W_{0}$ is a weight matrix, $m_{0}(y_{0},x,\psi)$ is a moment function induced from the information for the control condition in setup (a) or from the population information in setup (b), which is set so that $\E{m_{0}(y_{0},x,\psi)|z=0}=0$. For example of the case with $k_{y_{0}}(y_{0}) + k_{x}(x)= \beta_{y0} y_{0} + x'\beta_{x}$, the moment function is
\begin{align*}
m_{0}(x,y_{0},\psi)=
\begin{pmatrix}
1/p(z=0|y_{0},x,\psi)-1/p(z=0) \\
(x-\E{x})/p(z=0|y_{0},x,\psi) \\
(y_{0}-\E{y_{0}})/p(z=0|y_{0},x,\psi) \\
\end{pmatrix}.
\end{align*}

\subsection{MCMC implementation}

For both setups (a) and (b) with micro-level data, we use a straightforward application of traditional Markov chain Monte Carlo (MCMC) algorithms to draw samples of parameter vector $\psi$ and missing $y_{0}$'s for $r=1$ and $z=1$ from their joint posterior distribution.
We employ a data augmentation approach in which the missing $y_{0}$'s for $r=1$ and $z=0$ are drawn from the quasi-joint posterior distribution. In this setup, the missing $y_{i0}$ value for $r_{i}=1$ and $z_{i}=1$ is drawn from the full conditional posterior distribution, which is proportional to $p(z_{i}=1|y_{i0},x_{i},\psi) p(y_{i1},y_{i0}|x_{i},\psi)$. To draw samples from the full posterior distribution of $\psi$, it is convenient to use the Metropolis-Hastings algorithm, through which we draw the new candidate $\psi^{cand}$ from the density proportional to $L(\psi| y^{\obs}, y^{mis}, z^{\obs} z, r, x) \times p(\psi)$, where
\begin{align*}
&L(\psi| y^{\obs}, y^{mis}, z^{\obs}, r, x) \\
&=\prod_{i:r_{i}=1}^{N} \left[ p(y_{i1},y_{i0}|x_{i},\psi)p(z_{i}=1|y_{i0},x_{i},\psi) \right]^{z_{i}} \times \left[ p(y_{i0}|x_{i},\psi)p(z_{i}=0|y_{i0},x_{i},\psi) \right]^{1-z_{i}} ,
\end{align*}
and accept the new candidate with the probability
\begin{align}
min \Bigl\{ 1, \frac{\exp(Q_{0}(\psi^{cand})) }{\exp(Q_{0}(\psi^{old}))} \Bigr\}
\end{align}
where $\psi^{old}$ is the previous value (see also Igari and Hoshino, 2018).

\section{Simulation Study}

We conduct a simple simulation study to examine the performance of the estimator shown in the previous section. The data are generated as follows:
\begin{align*}
x \sim \normal{0}{1.5^{2}}, \ y_{0}\,|\,x \sim \normal{\mu_{0}(x)}{\sigma_{0}^{2}}, \ y_{1}\,|\,y_{0},x \sim \normal{\mu_{1}(y_{0},x)}{\sigma_{1}^{2}}
\end{align*}
where $\mu_{0}(x) = \theta_{00}+\theta_{01}x\,((\theta_{00},\theta_{01})=(1.0,0.6))$, $\mu_{1}(y_{0},x) = \theta_{10}+\theta_{11}x+\theta_{12}y_{0}+\theta_{13}y_{0}^{2}\,((\theta_{10},\theta_{11},\theta_{12},\theta_{13})=(1.5,0.5,0.6,-0.2))$, $\sigma_{0}=0.5$, $\sigma_{1}=0.6$, and the sample size is $N=1000$. We specify the missing mechanism (eq.~(\ref{eq:hirano_result})) as
$k_{0}=\beta_{0}$, $k_{x}(x) = \beta_{1} x$, and $k_{y_{0}}(y_{0}) =\beta_{2} y_{0}$ ($(\beta_{0},\beta_{1},\beta_{2})=(-1.2,0.8,0.6)$), 
and $z_{i}$ is drawn from $\mathrm{B}(1,p_{i})$ where $p_{i} = g(k_{0}+k_{y_{0}}(y_{i0}) + k_{x}(x_{i}))$ and $g$ is the logistic function. MCMC samples are obtained by the Hamiltonian Monte Carlo algorithm using Stan (Gelman et al., 2013) from the posterior, eq.~(\ref{eq:posterior}), which has 3000 iterations including 1000 for warm-up. 
\begin{figure}	[htbp]
	\centering
	\subcaptionbox{Plot of $(x,y_{0})$\label{fig:plot_x_y_{0}}}
	{\includegraphics[width=0.40\linewidth]{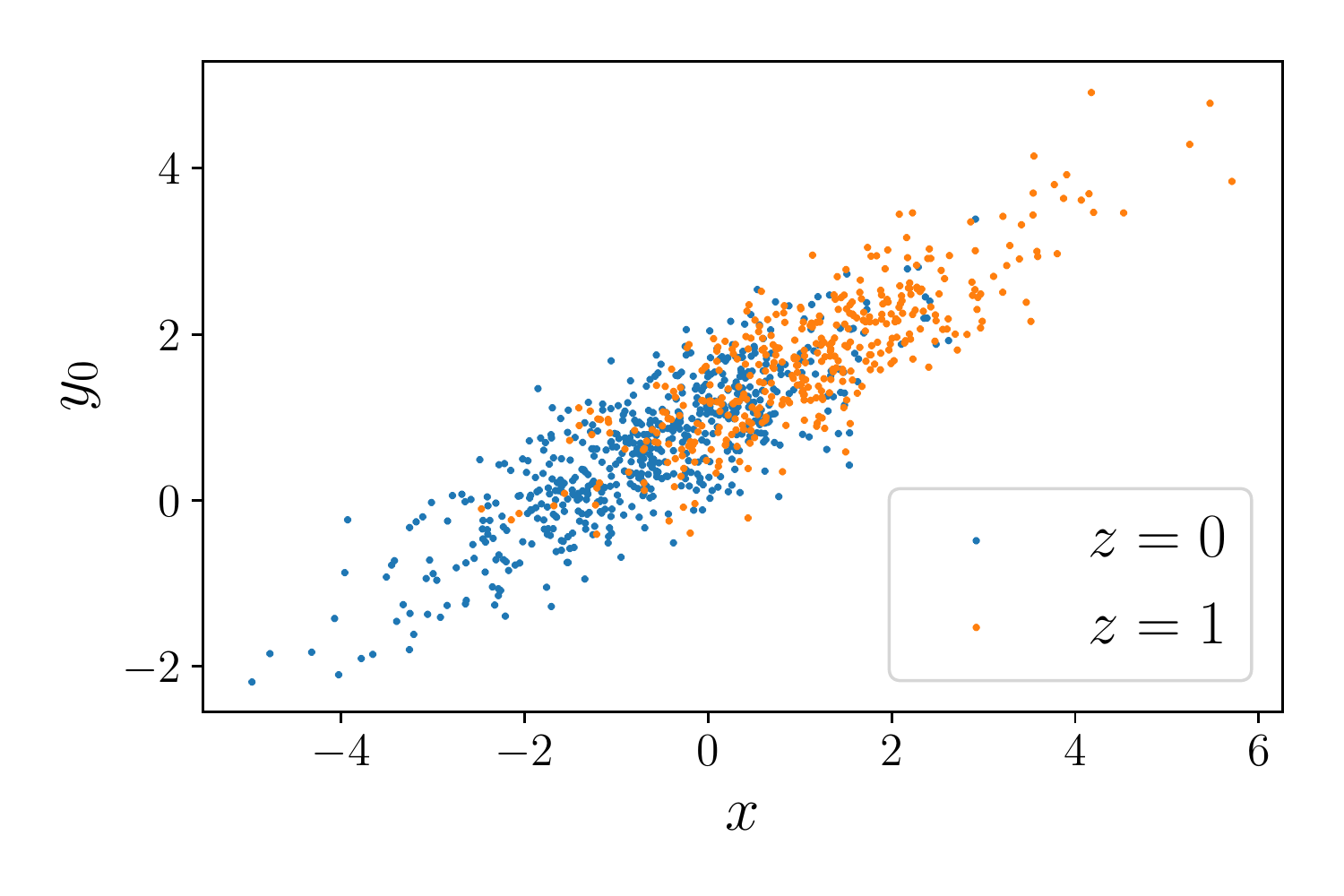}}
	\subcaptionbox{Plot of $(x,y_{1})$\label{fig:plot_x_y_{1}}}
	{\includegraphics[width=0.40\linewidth]{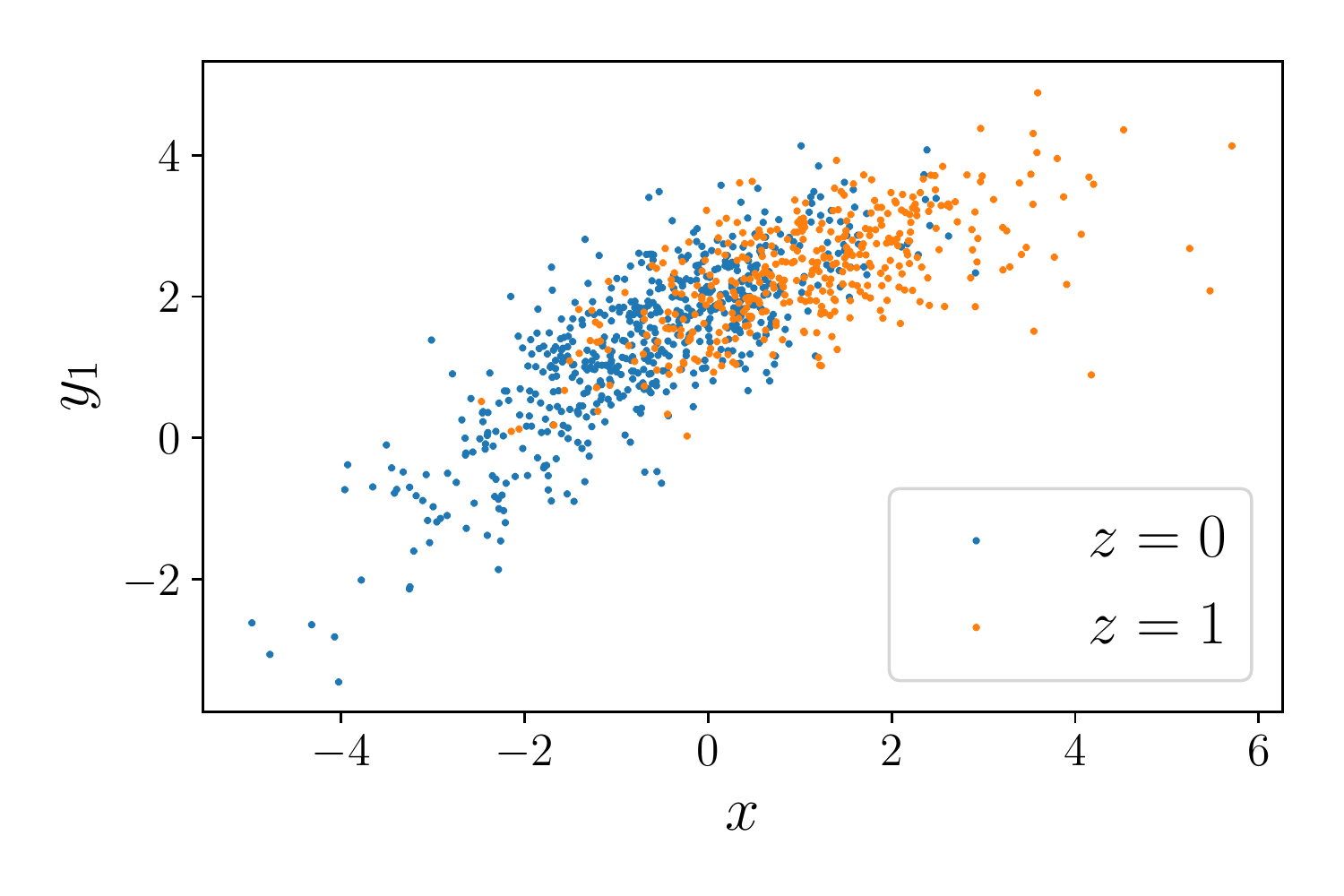}}	
	\caption{Scatter plots of the outcomes between the treatment and control group.}
	\label{fig:plots_y_dist} 
\end{figure}

Table \ref{table:result_summary} summarizes the MCMC samples over 500 replications. The mean of the posterior means was close to the true values for all the parameters, while the standard deviations were rather large, particularly in the parameters for the missing mechanism, $\beta$. Although the estimator of the missing mechanism varied greatly, it seemed not to have a significant impact on the other estimators, since the parameters for $p(y_{0}|x)$ were estimated with small variance.
\begin{table}[htbp]
\begin{center}
\caption{Summary of the simulation}
\begin{tabular}{lrccc}
\hline
Parameter (True) & Mean & s.d. & Cov. (\%) & MSE \\\hline
$\theta_{10} \, \, (1.5)$ & 1.509 & 0.371 & 95.6 & 0.242 \\
$\theta_{11} \,\, (0.5)$ & 0.496 & 0.161 & 95.2 & 0.047 \\
$\theta_{12} \,\, (0.6)$ & 0.576 & 0.353 & 94.8 & 0.224 \\
$\theta_{13} \,\, (-0.2)$ & $-0.191$ & 0.050 & 92.8 & 0.005 \\
$\theta_{00} \,\, (1.0)$ & 1.002 & 0.063 & 98.4 & 0.006 \\
$\theta_{01} \,\, (0.6)$ & 0.600 & 0.030 & 98.4 & 0.001 \\
$\sigma_{1} \,\, (0.6)$ & 0.574 & 0.050 & 97.2 & 0.005 \\
$\sigma_{0} \,\, (0.5)$ & 0.505 & 0.017 & 96.0 & 0.001 \\
$\beta_{0} \,\, (-1.2)$ & $-1.342$ & 0.725 & 98.0 & 0.869 \\
$\beta_{1} \,\, (0.8)$ & 0.806 & 0.383 & 98.0 & 0.230 \\
$\beta_{2} \,\, (0.6)$ & 0.677 & 0.675 & 98.4 & 0.728 \\\hline
\end{tabular}
\label{table:result_summary}
\end{center}
(Mean): mean of the posterior mean over 500 replications, (s.d.): standard deviation, (Cov): the proportion of times the 95\% credible interval of the estimated posterior mean contains the true value, (MSE): mean squared error over 500 replications.
\end{table}

Figure~\ref{fig:HTE_plot} shows the mean and 95\% credible bands for $\Ec{y_{1}}{y_{0}}$ over 500 replications. The red dashed line shows the true curve, while the black line shows the mean of the posterior means. As this indicates, the mean was very close to the truth and the credible band included it, capturing the nonlinear relationship between the potential outcomes. The reason why the variance is large on the right edge is that there is only the small number $y_{0}$ observed around there (see Figure~\ref{fig:plot_x_y_{0}}).
\begin{figure}[htbp]
	\centering
	\includegraphics[width=0.6\linewidth]{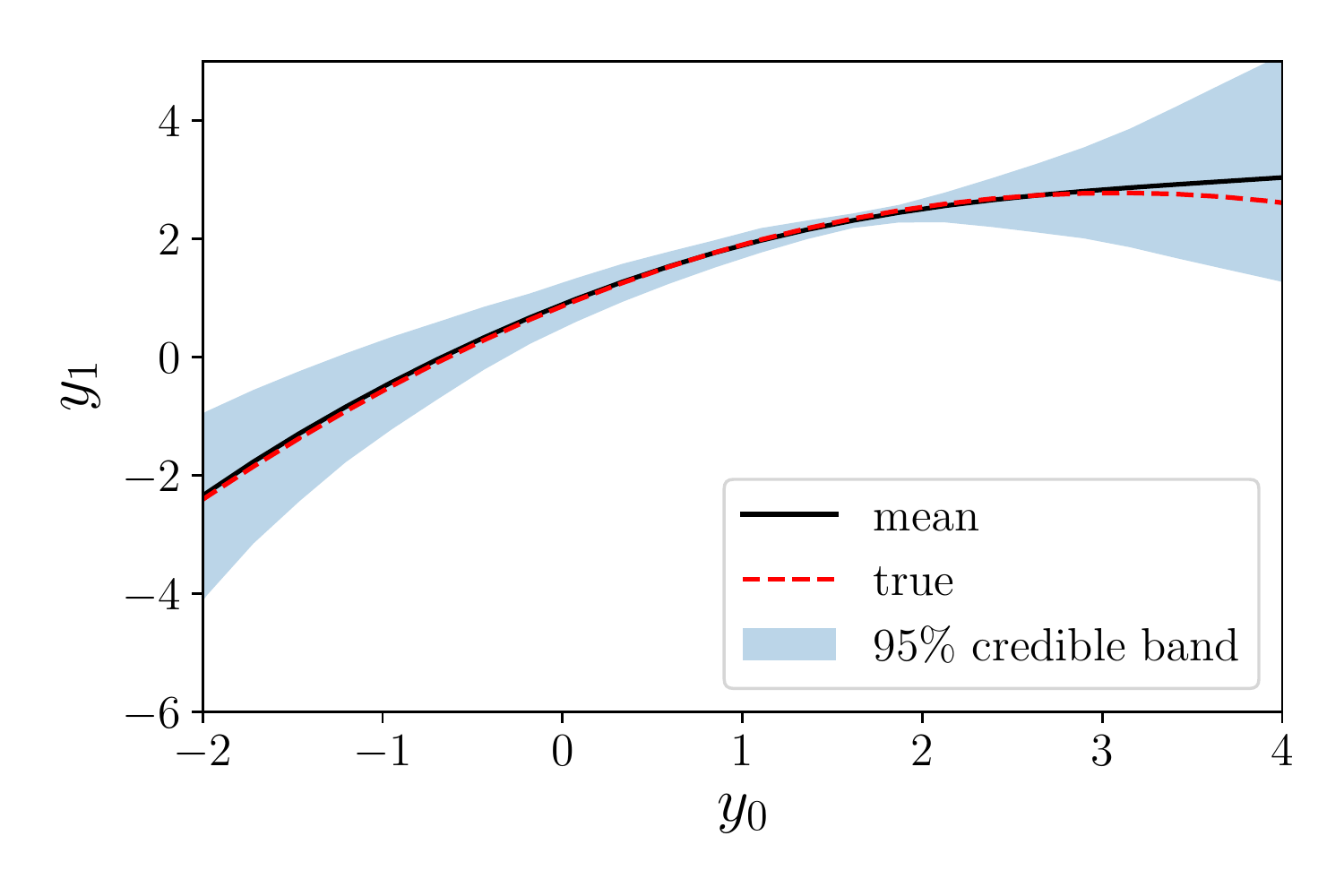}
	\caption{Result of the estimation of HTE: the mean and 95\% credible band}
	\label{fig:HTE_plot} 
\end{figure}

Table \ref{table:ate} compares the proposed estimate (``Prop") with the two existing methods, the simple mean difference (``Mean Diff.") and the inverse probability weighted estimator (``IPWE"), for the estimation of ATE. We calculated the proposed estimator through the following formula:
\begin{align*}
\E{y_{1}} &= \int \Ec{y_{1}}{y_{0},x,\psi}p(y_{0}|x,\psi)p(x)dy_{0}dx \\
&\simeq \frac{1}{N}\sum_{i=1}^{N}\int \Ec{y_{1}}{y_{0},x_{i},\psi}p(y_{0}|x_{i},\psi)dy_{0}.
\end{align*}
We also calculated the IPWE for ATE under strong ignorability, where we specified the propensity score as $p(z=1|x) = 1/(1+\exp(-(\gamma_{0}+\gamma_{1}x+\gamma_{2}x^{2})))$.

In this model setup, the average AUC (area under the curve) of the ROC (receiver operating characteristic), often called the ``c-statistic", over the generated 500 datasets, which is used a measure of the fit of a propensity score model using only observed covariates, was 0.837.
However, the result for IPWE under strong ignorability was biased, although the average c-statistic was greater than the recommended rule-of-thumb for propensity score methods, 0.8 (e.g. Ohman et al., 2008).
On the other hand, the mean of the results from our model was close to the truth, with high coverage rates.
The LATE estimator is not the estimator of ATE, but the estimator of the average treatment effect ``for the complier", so it is reasonable that the average of LATE estimates (0.641) is considerably different from the true value of ATE. In short, the LATE estimator should not be treated as the ATE estimator.

\begin{table}[htbp]
\centering
\caption{Estimation of ATE over 500 replications}
\begin{tabular}{lrccc}
\hline
ATE (True: 0.688) & Mean & s.d. & Cov. (\%) & MSE \\\hline
Prop. & 0.693 & 0.0925 & 98.3 & 0.0146 \\
IPWE & 0.727 & 0.106 & 78.0 & 0.0429 \\
Mean Diff. & 0.260 & 0.015 & 0.0 & 0.1860 \\\hline
\end{tabular}
\label{table:ate}
\end{table}

\section{Real Data Analysis}

We apply our method to the National Job Training Partnership Act (JTPA) study (Bloom et al., 1997; Abadie et al., 2002), which is one of the largest randomized job-training evaluations in the US. In the JTPA, the data of about 20,000 participants are collected from 16 regions for evaluation. These areas were not randomly chosen from all of the areas where the JTPA services are delivered; instead, the treatment was randomly assigned to participants within sites. The JTPA study is suitable for the purpose of our analysis because it captures a large amount of one-sided noncompliance; although the treatment was randomly assigned, only about 60\% of the participants in the treatment group received JTPA services, while about 2\% of the participants in the control group did. Thus, treatment status is likely to be dependent on potential outcomes, satisfying weak ignorability condition.

In this illustrative analysis, we focus on the target outcome and participants, following Abadie et al.~(2002). As the target outcome, we evaluate the sum of earnings over 30 months after random assignment; for the target population, we focus on adult women. The covariates used in the analysis are earnings in the previous year, age, dummies for black or Hispanic applicants, a dummy for high-school graduates (including GED (General Educational Development)-holders), dummies for married applicants, and whether the applicant worked at least 12 weeks in the 12 months preceding random assignment. Then, $d=Dim(x)=$ 6. In Tables 1 in the supplementary material, we reproduce the descriptive statistics reported in Abadie et al. (2002, Table 1) (the data are available on Angrist's web page). Note that we ignore the participants who were assigned to the control group but received JTPA services (about 2 \% of participants in the control group) to focus on one-sided noncompliance, so some parts of the table are different from that of Abadie et al.~(2002).
Because the outcome considered here is wage, which is censored at zero due to unemployment at follow-up period, we employ a variant of the Tobit-type model (for the Tobit-type model, see e.g. Amemiya 1985; Koop et al. 2007); we assume the latent variables behind the potential outcomes,
\begin{align*}
y_{0}=\begin{cases}
y_{0}^{*} & (y_{0}^{*}>0)\\
0 & (y_{0}^{*}\leq0)
\end{cases},\quad y_{1}=\begin{cases}
y_{1}^{*} & (y_{1}^{*}>0)\\
0 & (y_{1}^{*}\leq0)
\end{cases},
\end{align*}
and $p(y_{0}^{*}|x)$ and $p(y_{1}^{*}|y_{0}^{*},x)$ are the Gumbel and normal distributions respectively:
\begin{align*}
y_{0}^{*} & \sim G(\mu_{0}(x),\sigma_{0}^{2}),\quad y_{1}^{*}
\sim \normal{\mu_{1}(y_{0}^{*},x)}{\sigma_{1}^{2}},
\end{align*}
where $\mu_{0}(x)=\xi_{0}+x'\xi_{x}$ and $\mu_{1}(y_{0}^{*},x)=\lambda_{0}+\lambda_{1}y_{0}^{*}+x'\lambda_{x}$, and $G(a,b^{2})$ is the Gumbel distribution where the location and scale parameters are $a$ and $b$. Although it is standard to use a normal distribution for the Tobit model, we employ the Gumbel distribution, because we observe skewness in the distribution of $y_{0}^{*}$ and in its computational stability. The density of $y_{0}$ in the control group decreases from zero to the tail, while the censored units accounts for only about 14\% of the population, so a symmetrical distribution may be inappropriate with censoring considered\footnote{We also tried to use the skew-normal distribution for $p(y_{0}^{*})$, but the MCMC algorithm did not converge in this case.}. There are several studies where the Gumbel distribution is used for analyses of earnings (e.g. Horsky, 1990). See supplementary material for a histogram of $y_{0}$ for $r=0$ and the fitted Gumbel density.
Finally, we specify compliance probability as $p(z=1|y_{0}^{*},x)=1/(1+\exp(-(\beta_{0}+\beta_{1}y_{0}^{*}+\beta_{2}y_{0}^{*2}+x'\beta_{x})))$.
Then, the counterpart of eq.~(\ref{eq:int_eq}) is 
 $p(y_{1}^{*}|x,z=1)=\int p(y_{1}^{*}|y_{0}^{*},x)p(y_{0}^{*}|x,z=1)dy_{0}^{*}.$

In this setting, we need to make some modifications to the estimation method for compliance probability, $p(z=1|y_{0}^{*},x)$, described in section 3. First, we cannot use $p(y_{0})$ directly from the control group to identify $p(z=1|y_{0}^{*},x)$, so we need to estimate $p(y_{0}^{*})$ by $\hat{p}(y_{0}^{*})\simeq1/N_{\text{c}}\sum_{i:r_{i}=0}\hat{p}(y_{0}^{*}|x_{i})$ for moment conditions. Second, since an exact value of $y_{0}^{*}$ is not observed if $y_{0}^{*}\leq0$, we need to sample censored $y_{0}^{*}$'s as the parameter during iterations of MCMC. 

The posterior probability is then
\begin{align*}
p(\psi,\bar{y}_{0}^{*}|y^{\obs}, z, r, x) & \propto\prod_{i:r_{i}=1,z_{i}=1}\left[\left(\int p(y_{i1}^{*}|y_{i0}^{*},x_{i},\psi)p(y_{i0}^{*}|x_{i},\psi)p(z_{i}=1|y_{i0}^{*},x_{i},\psi)dy_{i0}^{*}\right)^{\delta(y_{i1}>0)}\right.\\
 & \left.\qquad\times\left(\iint_{-\infty}^{0}p(y_{i1}^{*}|y_{i0}^{*},x,\psi)p(y_{i0}^{*}|x,\psi)p(z_{i}=1|y_{i0}^{*},x_{i},\psi)dy_{i1}^{*}dy_{i0}^{*}\right)^{\delta(y_{i1}=0)}\right]\\
 & \quad\times\left. \prod_{i:r_{i}=1,z_{i}=0}\right[\left(p(y_{i0}^{*}|x_{i},\psi)p(z_{i}=0|y_{i0}^{*},x_{i},\psi)\right)^{\delta(y_{i0}>0)}\\
 & \left.\qquad\times\left(\int_{-\infty}^{0}p(y_{i0}^{*}|x_{i},\psi)p(z_{i}=0|y_{i0}^{*},x_{i},\psi)dy_{i0}^{*}\right)^{\delta(y_{i0}=0)}\right]\\
 & \quad\times\left[\prod_{i:r_{i}=0}\left(p(y_{i0}^{*}|x_{i},\psi)\right)^{\delta(y_{i0}>0)}\left(p(\bar{y}_{i0}^{*}|x_{i},\psi)\right)^{\delta(y_{i0}=0)}\right]\\
 & \quad\times\exp(Q_{0}(\psi)).
\end{align*}
where $\psi$ is the parameter of interest, $\psi=(\xi,\lambda,\sigma,\beta)$, $\delta(\cdot)$ is an indicator function, $Q_{0}(\cdot)$ is the moment constraint as eq.~(\ref{eq:obj_gmm}), given from
\begin{align*}
m_{0}(x,y_{0}^{*},\psi)=\begin{pmatrix}1/p(z=0|y_{0}^{*},x,\psi)-1/p(z=0)\\
(x_{1}-\E{x_{1}})/p(z=0|y_{0}^{*},x,\psi)\\
\vdots\\
(x_{d}-\E{x_{d}})/p(z=0|y_{0}^{*},x,\psi)\\
(y_{0}^{*}-\E{y_{0}^{*}})/p(z=0|y_{0}^{*},x,\psi)\\
(y_{0}^{*2}-\E{y_{0}^{*2}})/p(z=0|y_{0}^{*},x,\psi)
\end{pmatrix},
\end{align*}
and $\bar{y}_{0}^{*}=(\bar{y}_{10}^{*},\dots,\bar{y}_{N_{\text{c}}0}^{*})$
is the vector of latent variables for the units of zero-earnings in
the control group.

The estimated coefficients are shown in Table \ref{table:jtpa_coef_women}, where the variables with respect to earnings are scaled by $1/10000$. For compliance probability $p(z=1|y_{0}^{*},x)$, it can be seen that the latent variable for the potential untreated outcome, $y_{0}^{*}$, has a significant impact on compliance probability, which implies that weak ignorability condition is satisfied. While some of the other variables may not be significant, because the credible intervals include zero, it seems that educational background (whether a participant has a high school diploma or GED) has a positive impact on compliance probability. For the conditional mean $\E{y_{1}^{*}|y_{0}^{*},x}=\mu_{1}(y_{0}^{*},x)$, earnings in the past year have a positive effect on the potential treated outcome, while the dummy for black or Hispanic applicants is negative, consistent with most existing literature.
\begin{table}
\caption{Estimated Coefficients}
\centering
{\footnotesize
\begin{tabular}{cccccccc}
\hline 
 & \multicolumn{3}{c}{$\mu_{1}(y_{0}^{*},x)$} & & \multicolumn{3}{c}{$p(z=1|y_{0}^{*},x)$}\tabularnewline
\cline{2-4} \cline{6-8} 
Coefficient & mean & std & 95\% interval & & mean & std & 95\% interval\tabularnewline
\hline 
$y_{0}^{*}$ & $1.12$ & $0.04$ & $[1.04,1.20]$ & & $0.62$ & $0.11$ & $[0.42,0.85]$\tabularnewline
$y_{0}^{*2}$ & - & - & - & & $-0.16$ & $0.03$ & $[-0.22,-0.11]$\tabularnewline
Earnings in past year& $0.36$ & $0.12$ & $[0.12,0.59]$ & & $-0.37$ & $0.14$ & $[-0.64,-0.09]$\tabularnewline
Age & $0.18$ & $0.17$ & $[-0.15,0.52]$ & & $-0.28$ & $0.20$ & $[-0.70,0.13]$\tabularnewline
Married & $0.12$ & $0.08$ & $[-0.02,0.26]$ & & $0.06$ & $0.08$ & $[-0.10,0.23]$\tabularnewline
HS or GED & $0.11$ & $0.07$ & $[-0.03,0.25]$ & & $0.32$ & $0.08$ & $[0.16,0.48]$\tabularnewline
Black or Hispanic & $-0.16$ & $0.07$ & $[-0.29,-0.03]$ & & $-0.05$ & $0.06$ & $[-0.18,0.08]$\tabularnewline
{\scriptsize Work less than 13 weeks in past year} & $0.03$ & $0.08$ & $[-0.12,0.20]$ & & $-0.14$ & $0.09$ & $[-0.33,0.02]$\tabularnewline
\hline 
\end{tabular}
}
\label{table:jtpa_coef_women}
\end{table}

The latent variable $y_{0}^{*}$ has a positive effect on the conditional mean $\mu_{1}(y_{0}^{*},x)$, which implies that $y_{1}^{*}$ increases linearly according to $y_{0}^{*}$. However, in terms of the potential outcomes, $y_{0}$ and $y_{1}$, the linear relationship changes due to censoring at $y_{0}^{*}=y_{1}^{*}=0$. Figure~\ref{fig:jtpa_hte_women} shows the estimated HTE along with the estimated density of $p(y_{0}|y_{0}>0)$ at the bottom\footnote{See the supplementary material for the detailed derivation of the HTE shown here}. The red line and the blue domain are the posterior mean and the 95\% credible band, respectively; it makes sense that the credible band becomes larger as the density gets smaller. In addition, the HTE is discontinuous at $y_{0}=0$, because the integrated value of $\E{y_{1}^{*}|y_{0}^{*},x}$ over $y_{0}^{*}<0$ is shown at $y_{0}=0$; thus, its posterior mean and 95\% credible interval are shown as a point and interval on line $y_{0}=0$.
\begin{figure}[htbp]
	\centering
	\includegraphics[width=0.6\linewidth]{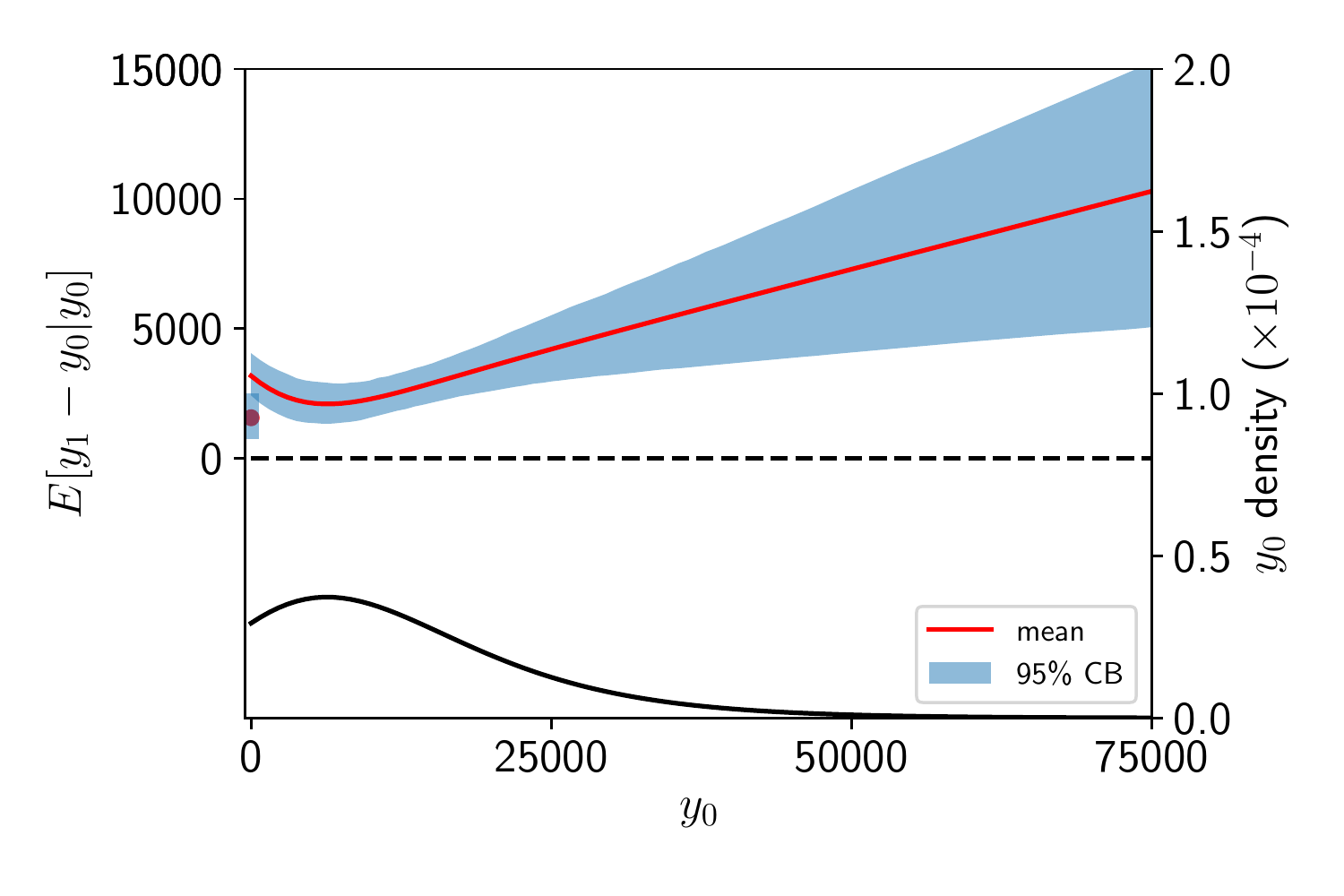}
	\caption{Estimated Heterogeneous Treatment Effects: the posterior mean and 95\% credible band}
	\label{fig:jtpa_hte_women} 
\end{figure}

The result is consistent with Abadie et al. (2002) in that, for women, the JTPA programme worked most effectively for the low-earning participants in terms of the proportion of the treatment effect to the untreated outcome. In addition, although $\mu_{1}(y_{0}^{*},x)$ is a linear function of $y_{0}^*$ (due to the nature of Tobit-type modelling, which is truncated at zero for both $y_{1}$ and $y_{0}$: see supplementary file for details), our analysis also suggests that the programme effect changes like a quadratic function over the untreated outcome; that is, the results indicate that the programme has an impact on those who are unemployed or with high earnings, but not so much those with moderate earnings (see also Table \ref{table:hte_quantile} which shows the HTE as a proportion of $y_{0}$ evaluated at several quantiles).
\begin{table}
\caption{Intervals and the HTE}
\centering
\begin{tabular}{cccccc}
\hline 
Interval & below 0 & 0 to 25 \% & 25 \% to 50 \% & 50 \% to 75 \% & above 75 \% \tabularnewline
\hline
$\E{y_{0}}$ & 0 & 548 & 4685 & 13522 & 29946 \tabularnewline
$\E{y_{1}-y_{0}}$ & 1568 & 2986 & 2260 & 2726 & 4817 \tabularnewline
$\E{y_{1}-y_{0}}/\E{y_{0}}$ (\%) & - & 544 & 48 & 20 & 16 \tabularnewline
\hline 
\end{tabular}
\label{table:hte_quantile}
\end{table}

The estimated causal effects are summarized in Table \ref{table:jtpa_causal_effects}. Note that ATE is estimated based on the model, while the LATE estimate is an IV (Wald) estimator. In addition, while the standard deviation and 95\% credible band for ATE are obtained from the posterior, for LATE and the simple difference by training status they are calculated from the asymptotic distributions. The estimated ATE is larger than the estimated LATE, which indicates that by promoting or compelling participation in the training under the treatment, the benefits of the programme could be improved.
\begin{table}[htbp]
\caption{Estimated Causal Effects}
\centering
\begin{tabular}{cccc}
\hline 
 & mean & s.d. & 95\% interval\tabularnewline
\hline 
ATE & 2701 & 238 & {[}2238, 3167{]}\tabularnewline
LATE & 1916 & 38 & {[}1843, 1990{]}\tabularnewline
Mean Diff. & 1260 & 5 & {[}1249, 1270{]}\tabularnewline
\hline 
\end{tabular}
\label{table:jtpa_causal_effects}
\end{table}

\section{Discussion and Conclusion}

We provide the sufficient condition for the identification of HTE with information on the marginal distribution of the untreated outcome under the nonignorable missing assumption, the same result for ATE and ATU, and a weaker condition for ATT. 
We propose a Bayesian and quasi-Bayesian estimation method for HTE and examine its properties through a simple simulation study, showing the availability of estimating $\E{y_{1}|y_{0}}$ even though neither of the pair $(y_{1},y_{0})$ is observed. The proposed method is also applied to the analysis of the JTPA study, providing the consistent result with the previous literature and also several new insights.

The completeness condition, as used in the proof of Theorem 2, has received considerable research attention since it was first presented in Newey and Powell~(2003). However, although it is applied widely, Canay et al.~(2013) show that the completeness condition cannot be tested using observed data. This implies that ``for every complete distribution, there exists an incomplete distribution which is arbitrarily close to it" (Freyberger, 2017, p.1629). Canay et al. (2013) argue that their result does not suggest avoiding the use of the completeness condition but rather justifies it through alternative arguments. In contrast, several studies provide sufficient conditions as alternatives for the completeness condition that may be testable (Newey and Powell, 2003; D'Haultfoeuille, 2011; Hu and Shiu, 2018). Our assumption that the function $g$ is the logistic function relates to the latter approach. Newey and Powell~(2003) provide the sufficient condition that a certain conditional distribution, corresponding to $p(y_{0}|x,z=1)$ in our model, is of the exponential family; however, specifying the extended propensity score as the logistic regression may be somewhat weaker.

Our results relate to the partial identification literature on statistical data fusion or statistical data combination (see e.g. Ridder and Moffitt, 2007). Fan et al.~(2014) consider a situation where the outcome variables and covariates are separately observed and derive partial identification results. Although they assume strong ignorability (i.e. that the missing mechanism is ignorable), no sample is observed as a set of the outcome and covariates. In contrast, we assume that the outcome and covariates are observed simultaneously in each group; hence, in this sense, their setting is more general than ours. However, we consider the nonignorable missing mechanism (weak ignorability) and provide point identification results using auxiliary information. Therefore, which approach is the more useful may depend on the situation.
In addition, as a practical application to data fusion, our results indicate that if we know moments of variables of both data sets to be combined, then the analysis can be done without several strong assumptions often assumed in the literature (Takahata and Hoshino, 2018).

While the theorems here deal with an semiparametric identification, the proposed estimation methods consider parametric models. Nonparametric or semiparametric Bayesian estimation methods (see e.g. Dunson et al., 2007; Ghosal and van der Vaart, 2017; Hoshino, 2013) will be applied to modeling $p(y_{1}|y_{0},x)$ and $p(y_{0}|x)$. In addition, we have developed some estimation methods for nonparametric models that may be applied toward practical applications (see, Takahata and Hoshino, 2018).
However, we believe this paper is the first milestone research on this issue, because no previous study has found sufficient conditions for identification of $p(y_{1}|y_{0})$ which are weak enough to be assumed in many empirical applications.

As we mentioned, we can estimate HTE, ATE, ATT, and ATU for a randomized controlled trial with one-sided non-compliance, that is, where in the control group all the participants comply with the control condition while in the treatment group not all the participants comply with their treatment or individuals are allowed to choose their treatment. The RCT design with one-sided non-compliance is very easy to implement, because researchers do not enforce complete compliance, which will assure ecological validity (Shadish, Cook and Campbell, 2002) while the proposed estimation methods can assure internal validity.

\end{document}